\documentclass{article}
\usepackage{mathrsfs}
\usepackage{a4wide}
\usepackage{amsmath,amsfonts,amssymb,amsthm}
\usepackage{graphicx}

\def \beq{\begin{equation}}
\def \eeq{\end{equation}}

\def\and {{\rm \; and \;}}

\newtheorem{theorem}{Theorem}[section]

\newtheorem{lemma}[theorem]{Lemma}

\theoremstyle{definition}
\newtheorem{remark}[theorem]{Remark}

\numberwithin{equation}{section}

\begin{document}

\noindent 
\begin{center}
\textbf{\Large On the cotunneling regime of interacting quantum dots}
\end{center}

\noindent

\begin{center}

\small{Horia D. Cornean\footnote{Department of Mathematical Sciences, Aalborg
     University, Fredrik Bajers Vej 7G, 9220 Aalborg, Denmark; e-mail:
     cornean@math.aau.dk},
Valeriu Moldovean${\rm u}^{1,}$\footnote{National Institute of Materials Physics, P.O. Box MG-7  Bucharest-Magurele, Romania;
e-mail: valim@infim.ro}}

\end{center}

\vspace{0.5cm}

\begin{abstract}

Consider a bunch of interacting electrons confined in a quantum dot. The later is 
suddenly coupled to semi-infinite biased leads at an initial instant $t=0$. We identify 
the dominant contribution to the ergodic current in the off-resonant transport regime, 
in which the discrete spectrum of the quantum dot is well separated from the absolutely 
continuous spectrum of the leads. Our approach allows for arbitrary strength of the 
electron-electron interaction while the current is expanded in even powers of the (weak) 
lead-dot hopping constant $\tau$. We provide explicit calculations for sequential tunneling 
and cotunneling contributions to the current. In the interacting case it turns out that the cotunneling 
current depends on the initial many-body configuration of the sample, while in the non-interacting case
it does not, and coincides with the first term in the expansion of the Landauer formula w.r.t $\tau$.    

\end{abstract}

\section{Introduction}\label{intro}

The dominant role of electron-electron interaction at mesoscopic scale has long been recognized, 
effects like Coulomb blockade, Kondo correlations or charge sensing being currently observed and 
even manipulated in transport experiments. The typical system consists of a few-level quantum dot
coupled to source and drain probes (leads).

In the physics comunity, different approaches to 
the transport problem in interacting systems were developed and intensively used 
for numerical simulations. The choice of the method depends on which parameter of the 
problem allows for a perturbative treatment. If the interaction strength $U$ is rather small, one 
can use the non-equilibrium Green-Keldysh formalism to compute transient or 
steady-state currents by approximating the interaction effects at
different levels \cite{Myohanen}. In the strongly interacting case,
two alternative methods are available: the first one is the T-matrix approach
\cite{Flensberg}, and the second one is the generalized Master
equation formalism \cite{Timm}. Both methods 
rely on perturbative expansions w.r.t to the parameter $\tau$ which
measures the coupling between the leads and the 
sample, while $U$ is typically much larger that $\tau$.  

Compared with the richness of the physics literature on these
subjects, only few rigorous results exist on time-dependent transport in interacting systems,
 and they only apply to weakly interacting systems
\cite{Jaksic,CMP}. More precisely, one needs two important conditions in order to guarantee the existence of a stationary state (NESS): 
1. the single-particle Hamiltonian describing the non-interacting system has purely absolutely 
continuous spectrum and 2. the interaction strength is sufficiently small. Under these conditions, one 
can write down {\it exact} formulas for the stationary current, but the calculations must be performed perturbatively in the interaction. 
These results are valid both for the partitioning \cite{Caroli} and the 
partition-free \cite{Cini} transport scenarios; moreover \cite{CMP}, one can prove that 
in the partitioning case the stationary current is {\it independent} on the initial state 
of the sample.  

In this paper we deal with a strongly interacting regime which drastically differs from the one discussed above. 
In particular, the previously mentioned two conditions are not satisfied. More precisely, here we only consider 
quantum dots whose discrete spectrum is far away from the absolutely continuous 
spectrum of the leads. Otherwise stated, this {\it off-resonant condition} says that the single-particle Hamiltonian of the fully
coupled system has discrete bound states and no resonances. The existence of a stationary state in this case
is still an open problem and we do not address it here. We know though that even in the non-interacting case one must take the  
ergodic limit in order to kill off the bound state induced current oscillations \cite{Aschb, CNZ,Stefanucci}. 

From now on $\tau$ will denote the hopping constant between the leads and sample. If $I_{\alpha,t}(\tau)$ is the current at time $t\geq 0$ 
in a given lead $\alpha$, its ergodic (Ces{\`a}ro) limit is defined as:
\begin{equation}\label{Jerg}
I_{\alpha,\infty}(\tau):=\lim_{T\to\infty}\frac{1}{T}\int_0^{T} I_{\alpha,t}(\tau)dt\quad 
\left (=\lim_{\eta\searrow 0}\eta\int_0^{\infty}e^{-\eta t}I_{\alpha,t}(\tau)dt\right ),
\end{equation}
where the second equality expresses the known fact that if the Ces{\`a}ro limit exists, then it can also be calculated through the 
Abel limit $\eta\searrow 0$. 

The central object of our study will be the quantity  
\begin{equation}\label{Jeta1}
I_{\alpha}(\eta,\tau):=\eta\int_0^{\infty}e^{-\eta t}I_{\alpha,t}(\tau)dt,\quad \eta >0,
\end{equation}
and its behavior as a function of $\tau$. 
We will show that for a fixed $\eta>0$ and for $\tau/\eta$ sufficiently small one 
can expand the RHS of Eq.(\ref{Jeta1}) in a convergent series of even powers of $\tau$, that is:
\begin{equation}\label{Jeta}
I_{\alpha}(\eta,\tau)=\sum_{k=1}^{\infty}\tau^{2k}C_{\alpha,2k}(\eta).
\end{equation}
The main questions are the following: 
{\it \begin{enumerate}
\item  When does the Ces{\`a}ro limit $I_{\alpha,\infty}$ exist? (This would imply that  $I_{\alpha}(0_+,\tau)$ exists and equals $I_{\alpha,\infty}$).
\item How many coefficients $C_{\alpha,2k}(\eta)$ admit the limit $\eta\searrow 0$?
\item If $I_{\alpha,\infty}$ exists, does it have an asymptotic expansion around $\tau=0$? If yes, 
can we write:
\begin{equation}\label{Jeta2}
I_{\alpha,\infty}\sim \sum_{k=1}^{\infty}\tau^{2k}C_{\alpha,2k}(0_+)?
\end{equation}
\end{enumerate}
}
We have a good understanding of the problem for non-interacting systems. The off-resonant 
condition on the spectrum of the one particle Hamiltonian is crucial; in the resonant case we can prove that $I_{\alpha,\infty}$ exists and has an 
asymptotic expansion, but its leading coefficient is {\it not} given by $C_{\alpha,2}(0_+)$ (compare \eqref{LB33} with \eqref{seq}). 
The interacting case is open.

In this paper we will identify and compute the coefficients corresponding to $k=1,2$ for the interacting case under some off-resonant conditions. 
These two terms have a clear physical meaning. $C_{\alpha,2}$ describes the sequential tunneling processes 
(i.e. electrons enter or leave the dot one-by-one). 
In the off-resonant regime considered here this contribution is absent because the energy conservation requires 
some levels of the isolated dot to be within the continuous spectrum of the leads. $C_{\alpha,4}$ then gives the dominant 
contribution to the current and contains the so-called cotunneling processes, in which electrons tunnel from and to the dot in 
pairs (cooperative tunneling). To our best knowledge, the cotunneling regime has not been previously discussed in a rigorous context.   

Before summarising the content of the paper let us comment on the different transport regimes 
(i.e. resonant vs. cotunneling) and some subtleties related to the existence of NESS and of the perturbative 
expansion w.r.t. $\tau$. In the absence of the electron-electron interaction the steady-state current is given by the 
Landauer formula (see \eqref{LB}) 
which has been rigorously proved using various methods \cite{Aschb, CNZ, CJM1,Nenciu,CDNP-1}. This formula implies an effective resolvent 
which is {\it not} always analytic w.r.t $\tau$ (see the discussion around \eqref{calT}and \eqref{LB33}).

The paper is organized as follows: in Section \ref{rezultate} we introduce the model, the problem, state the main result, and give a number of 
consequences. Section \ref{limitatermo} deals with the thermodynamic limit while Section \ref{offresonant} 
contains the proof of the sequential and cotunneling formulas. Section \ref{numerical} is devoted to numerical results obtained via the 
generalized Master equation method \cite{MMG}. We find out that in the off-resonant case the 
current does not settle to a steady-state, but the ergodic limit seem to exist. We conclude in Section \ref{concluzii}.

\section{Notation, setting and the main result}\label{rezultate}

We shall adopt the partitioning approach to the transport problem \cite{Caroli}. A finite system $S$ is coupled to $M\geq 2$ 
 noninteracting one-dimensional semi-infinite leads (i.e. particle reservoirs) at some initial instant $t_0$. 
For simplicity, we consider a discrete model in which the sample is modelled as a finite lattice 
$\Gamma \subset \mathbb{Z}^{2}$ and 
the leads are described by one-dimensional discrete Laplacians on the half-line with Dirichlet boundary conditions. 
The one-particle Hilbert space is thus 
${\cal H}:= l^2(\mathbb{N}_1)\oplus\dots\oplus l^2(\mathbb{N}_M)\oplus l^2(\Gamma)=:{\cal H}_L\oplus{\cal H}_S$. 
We shall use the geometrical (standard) basis in ${\cal H}_L$, which is the set $\{|i_{\alpha}\rangle : i\geq 0,\; 1\leq\alpha\leq M\}$ where 
$i_{\alpha}$ means the $i$-th site of the lead $\alpha$. Similarly we have the basis $\{|m\rangle \}_{m\in \Gamma}$ 
for ${\cal H}_S$. We denote by $|m_{\alpha}\rangle $ the vector corresponding to the sample site to which the lead $\alpha$ is attached. 

The leads are suddenly coupled to the sample at $t=0$. Then for $t>0$ 
the single-particle Hamiltonian reads as:
\begin{equation}\label{Hmic}
h=h_S+h_L+h_T,
\end{equation} 
where  
\begin{align}\label{hS}
h_S&=\sum_{ m,n\in\Gamma }t_{mn}|m\rangle\langle n|,\\\label{hL}
h_L&=\sum_{\gamma=1}^M t_L\left (\sum_{i\geq 0}
|i_{\gamma}\rangle\langle (i+1)_{\gamma}|+\sum_{i\geq 1}
|i_{\gamma}\rangle\langle (i-1)_{\gamma}|\right )
:=\sum_{\gamma=1}^Mh_{\gamma},\\\label{hT}
h_T&=\tau\sum_{\gamma=1}^M(|0_{\gamma}\rangle\langle m_{\gamma}|+|m_{\gamma}\rangle\langle 0_{\gamma}|  ).
\end{align}
In the above equation $h_T$ is the so called tunneling Hamiltonian and $\tau$ is the coupling strength.
Here $\{t_{mn}\}_{m,n\in\Gamma}$ is any symmetric matrix and $t_L>0$ is the hopping constant of the leads.

We also introduce the eigenfunctions and eigenvalues of $h_S$ and the 
generalized eigenfunctions of $h_{\gamma}$:
\begin{equation}\label{geneigen}
h_S\phi_{\lambda}=e_{\lambda}\phi_{\lambda}, 
\quad h_{\gamma}\varphi_E^{\gamma}=E\varphi_E^{\gamma},          
\end{equation}
where $E$ is the energy associated to an electron propagating on leads with momentum $q\in (0,\pi)$ 
(the leads are identical). The explicit form of $\varphi_E^{\gamma}$ 
in a given site $i\geq 0$ of the lead $\gamma$ is taken to be:
\begin{equation}\label{semiinf}
\varphi_E^{\gamma}(j)=\frac{\sin [(j+1)q]}{\sqrt{\pi t_L \sin q} }, \quad E=2t_L\cos (q)\in [-2t_L,2t_L],\quad |\varphi_E^{\gamma}(0)|^2=
 \frac{\sqrt{1-\frac{E^2}{4t_L^2}}}{\pi t_L }.
\end{equation}
When the leads are finite and of length $\Lambda$, the lead spectrum is purely discrete and given by $\{\varepsilon_{q_{\gamma}}\}$ where $q$ now takes 
discrete values. A corresponding eigenfunction is denoted by $\varphi_{q_{\gamma}}$. The notation of the corresponding Hamiltonians is changed into 
$h_{\gamma}^{(\Lambda)}$, and we have:  
\begin{equation}
h_{\gamma}^{(\Lambda)}\varphi_{q_{\gamma}}=\varepsilon_{q_{\gamma}}\varphi_{q_{\gamma}}.
\end{equation}

We now formulate the transport problem in the language of second quantization (see \cite{Martin} for the 
standard procedures and notations). Let ${\cal F}={\cal F}_L\otimes{\cal F}_S$ be the Fock space constructed 
from the Hilbert space $\cal H$. The interaction of strength $U$ between electrons in the sample is given by 
the two-particle operator:
\begin{equation}\label{Vint}
V=\frac{U}{2}\sum_{m,n\in \Gamma }v(m-n)a^*(|m\rangle)a(|m\rangle )a^*(|n\rangle)a(|n\rangle),
\end{equation}
where $a^*(|m\rangle)$ and $a(|n\rangle)$ are creation and  annihilation operators in the sites $m,n$ and $v(m-n)$ is a 
pair potential which by assumption is bounded for $m=n$. 
These operators act in the  antisymmetrized Fock space ${\cal F}_S$ generated by $l^2(\Gamma)$. 
Similarly one defines creation and  annihilation operators in the leads, 
and then rewrites $h$ in the second quantization w.r.t the geometrical basis.
We use capital letters to denote the second quantized versions of the one particle operators: $H_a=d\Gamma (h_a),\,\, a\in\{S,L,T\}$. 
Then the total Hamiltonian  of the coupled and interacting system reads as follows:
\begin{equation}\label{Htot}
H=H_S+V+H_L+H_T=:H_0+H_T
\end{equation} 
 
The current operator in the lead $\alpha$ is introduced as the time derivative of the electron number 
operator $N_{\alpha}=\sum_{i\geq 0 }a^*(i_{\alpha})a(i_{\alpha}) $. Using the anticommutation relations 
one gets:
\begin{equation}\label{current2}
J_{\alpha}=-e{\dot N}_{\alpha}=-\frac{ie}{\hbar}[H,N_{\alpha}]=-\frac{ie}{\hbar}[H_T,N_{\alpha}]
=\frac{ie\tau }{\hbar}(a^*(|0_{\alpha}\rangle)a(|m_{\alpha}\rangle)-a^*(|m_{\alpha}\rangle)a(|0_{\alpha}\rangle)).
\end{equation} 
>From now on we adopt the convention $e=\hbar=1$. Note that the same form of $J_{\alpha}$ holds for leads
of finite length.

The different chemical potentials of the leads are $\mu=[\mu_1,\mu_2,...,\mu_M]$, and  
the inverse temperature $\beta>0$ is taken constant. 
The equilibrium sub-state of the leads is characterized by the following density matrix:
\begin{equation}
\rho^{(\Lambda)}_L:=
\Pi_{\gamma=1}^M\frac{e^{-\beta (H^{(\Lambda)}_{L,\gamma}-\mu_\gamma N_\gamma)}}{{\rm Tr}_{{\cal F}_L}
\{ e^{-\beta (H^{(\Lambda)}_{L,\gamma}-\mu_\gamma N_\gamma)} \}},
\end{equation}
which consists of a Gibbs state on each lead.

The initial density matrix of the sample $\rho_S$ can be any positive function of $H_S+V$, with trace one. 
For example, if at $t\leq 0$ the mesoscopic sample is empty, then we
have to take $\rho_S=|0,0,...\rangle\langle 0,0,...|$  where $|0,0,... \rangle$ is the vacuum state in 
${\cal F}_S$ written w.r.t the occupation number basis. But equally
well, one may also 
consider that the sample already contains a few interacting particles at
$t\leq 0$. Let us denote by $|\nu\rangle$ the eigenstates of $H_S+V$, and by $E_{\nu}$ 
its many-body energies
($(H_S+V)|\nu\rangle=E_{\nu}|\nu\rangle$). Without loss of generality,
we will take $\rho_S$ to be a pure state 
given by an initial many-body state (MBS)  
henceforth denoted by $\nu_0$. Thus $\rho_S=|\nu_0\rangle\langle \nu_0|$.

The main quantity we are interested in is the statistical average 
of the current operator on lead $\alpha$. To this end we introduce the statistical operator $\rho^{(\Lambda)}$ of the system
with finite leads. It solves the quantum Liouville equation for $t>0$ and is given by:
\begin{equation}\label{quantumliu}
\rho^{(\Lambda)}(t)=e^{-itH^{(\Lambda)}}\rho_0^{(\Lambda)}e^{itH^{(\Lambda)}},\quad 
\rho_0^{(\Lambda)}:=\rho^{(\Lambda)}_L\otimes\rho_S.
\end{equation}
If $B$ is an observable acting in the Fock space ${\cal F}$, we denote by $B(t):=e^{itH^{(\Lambda)}}B e^{-itH^{(\Lambda)}}$ its Heisenberg 
evolution. Then the average value of $B$ at time $t$ is defined as:
\begin{equation}\label{ref}
 \langle B (t)\rangle_{{\rm ref }}:=\lim_{\Lambda\to\infty} {\rm Tr}_{\cal F} \{  \rho^{(\Lambda)}(t) B \}= 
\lim_{\Lambda\to\infty} {\rm Tr}_{\cal F} \{  \rho_0^{(\Lambda)} B(t) \},
\end{equation}
whenever this limit exists. Then our results are summarised in the following theorem:
\newpage
\begin{theorem}Let $f(x)=1/(e^{\beta x}+1)$ be the Fermi function and $f_{\alpha}(E)=f(E-\mu_{\alpha})$. Let $\chi_L$ be 
the characteristic function of the interval $[-2t_L,2t_L]$. Then:

\noindent {\rm (i)}. The transient current $I_{\alpha,t}(\tau)$ in the lead $\alpha $ is given by 
\begin{equation}\label{currTDL}
I_{\alpha,t}(\tau):=\langle J_{\alpha}(t)\rangle_{{\rm ref }},\quad t\geq 0,
\end{equation}
and defines an entire function of $\tau$. 

\noindent {\rm (ii)}. Let $I_{\alpha}(\eta,\tau)=\eta\int_0^{\infty}e^{-\eta t}I_{\alpha,t}(\tau)dt$ as in \eqref{Jeta1}.   
Then one has (see \eqref{Jeta} and \eqref{semiinf}):
\begin{align}\nonumber
&C_{\alpha,{\rm
    seq}}:=C_{\alpha,2}(0_+)=\frac{2}{t_L}\sum_{\nu} \sqrt{1-\frac{(E_{\nu_0}-E_{\nu})^2}{4t_L^2}}\\
&
\times \left\lbrace [1-f_{\alpha}(E_{\nu_0}-E_{\nu})]|A_{\nu\nu_0 }|^2\chi_L(E_{\nu_0}-E_{\nu}) 
-f_{\alpha}(E_{\nu}-E_{\nu_0})|A^*_{\nu\nu_0 }|^2  \chi_L(E_{\nu}-E_{\nu_0})\right \}, \label{seq}
\end{align}
where:
\begin{equation}\label{Anu}
A^{\#}_{\nu_i\nu_j}(m)=\langle \nu_i,a^{\#}(m) \nu_j \rangle, \quad \#=*,\cdot . 
\end{equation}

\noindent{\rm (iii)}. Assume that the following two off-resonant conditions
are fullfiled: 

a). If $|\nu\rangle,|\nu'\rangle$ differ by one particle, then 
$E_{\nu}-E_{\nu'}\notin [-2t_L,2t_L]$; 

b).  If  $|\nu\rangle, |\nu'\rangle$ differ by two
particles, then $E_{\nu}-E_{\nu'}\notin [-4t_L,4t_L]$.

\noindent Then we have: 
\begin{equation}\label{cot1}
C_{\alpha,{\rm
    cot}}:=C_{\alpha,4}(0_+)=\frac{1}{\pi^2t_L^2}\sum_{\gamma}\int_{-2t_L}^{2t_L}dE \;
\left (1-\frac{E^2}{4t_L^2}\right )
\left ( {\cal P}_{\gamma\alpha}(E)-{\cal P}_{\alpha\gamma}(E)  \right),
\end{equation}
where ${\cal P}_{\gamma\alpha}$ is the cotunneling rate:
\begin{align}\label{J4fin}
&{\cal P}_{\gamma\alpha}(E)=\sum_{\nu,\nu',\nu''}\\
&\left\lbrace \chi_{L}(E-E_{\nu'}+E_{\nu_0})\frac{f_{\gamma}(E) [1-f_{\alpha}(E-E_{\nu'}+E_{\nu_0})]}
{ (E_{\nu}-E_{\nu_0}-E)(E_{\nu'}-E_{\nu''}-E)  }
A_{\nu_0\nu}(m_{\gamma})A_{\nu\nu'}^*(m_{\alpha})A_{\nu'\nu''}^*(m_{\gamma})A_{\nu''\nu_0}(m_{\alpha})
\right .\nonumber\\\nonumber
&-\chi_{L}(E+E_{\nu'}-E_{\nu_0})\frac{[1-f_{\alpha}(E)] f_{\gamma}(E+E_{\nu'}-E_{\nu_0} )}
{(E_{\nu}-E_{\nu_0}+E) (E_{\nu_0}-E_{\nu''}-E )}
A_{\nu_0\nu}^*(m_{\alpha})A_{\nu\nu'}(m_{\gamma})A_{\nu'\nu''}^*(m_{\gamma})A_{\nu''\nu_0}(m_{\alpha}) \\
\nonumber
&+\chi_{L}(E-E_{\nu'}+E_{\nu_0})\frac{f_{\alpha}(E)[1-f_{\gamma}(E-E_{\nu'}+E_{\nu_0} )]}
{ (E_{\nu}-E_{\nu_0}-E)(E_{\nu_0}-E_{\nu''}+E)  }
A_{\nu_0\nu}(m_{\alpha})A_{\nu\nu'}^*(m_{\gamma})A_{\nu'\nu''}(m_{\gamma})A^*_{\nu''\nu_0}(m_{\alpha})\\
&\left. -\chi_{L}(E+E_{\nu'}-E_{\nu_0})\frac{[1-f_{\gamma}(E)]f_{\alpha}(E+E_{\nu'}-E_{\nu_0} )}
{ (E_{\nu}-E_{\nu_0}+E)(E_{\nu'}-E_{\nu''}+E)  }
A_{\nu_0\nu}^*(m_{\gamma})A_{\nu\nu'}(m_{\alpha})A_{\nu'\nu''}(m_{\gamma})A^*_{\nu''\nu_0}(m_{\alpha})
\right\rbrace .\nonumber 
\end{align}

\end{theorem}

\vspace{0.5cm}

\begin{remark}\label{remarca22}
Provided that \eqref{Jeta2} holds true, if $\tau$ is sufficiently small then the ergodic current $I_{\alpha,\infty}$ 
should be well approximated by $\tau^2C_{\alpha,2}(0_+)+
\tau^4C_{\alpha,4}(0_+)$. Both terms describe tunneling processes from and into the dot. 
Note that $E_{\nu'}=E_{\nu_0}$ is allowed in the above sums, thus $\chi_{L}(E-E_{\nu'}+E_{\nu_0})$ has to be replaced by $1$ in those terms.
\end{remark}

\begin{remark}\label{remarca22'}In the expression of $C_{\alpha,2}(0_+)$, the factor 
$(1-f_{\alpha}(E))|A_{\nu\nu_0 }|^2$ is the tunneling probability from the dot to the leads of an
electron with energy $E$ (the corresponding state in the lead must be empty). Similarly, the second term of $C_{\alpha,2}(0_+)$ 
represents processes in which the many-body state of the dot changes by 'absorbing' one electron from the leads.
These processes are called {\it sequential}, as electrons  tunnel one by one.
It is clear that in the off-resonant regime (i.e. $E_{\nu}-E_{\nu_0}\notin [-2t_L:2t_L]$) the sequential tunneling
is suppressed and one has to go to the next term. Note that in the
resonant regime of the non-interacting case, this term cannot be recovered by expanding the
Landauer formula in powers of $\tau$ (see \eqref{LB33} for further details). The phenomenon which happens is 
well described by the following
toy example. Let $K$ be a constant either equal to $0$ or $1$. Define the
functions:
$$I(\eta,\tau;K):=(1-K)\frac{\eta \tau^2}{\eta +\tau^2}+\frac{\tau^4}{K+\eta+\tau^2}\arctan\left
  (\frac{1}{K+\eta+\tau^2}\right ).$$
The resonant case is modeled by the condition $K=0$. In that case we
have:
$$I(0_+,\tau;0)=\tau^2\arctan(1/\tau^2)=\tau^2\frac{\pi}{2}+\mathcal{O}(\tau^4),\quad I(\eta,\tau;0)=
\tau^2+\mathcal{O}(\tau^4),$$
which shows that $C_{\alpha,2}(0_+)=1$ and we cannot recover the
'true' behavior of  $I(0_+,\tau;0)$ from such an expansion. 

The off-resonant case is modeled by $K=1$. Then:
$$I(0_+,\tau;1)=\tau^4\frac{\pi}{4}+\mathcal{O}(\tau^6),\quad 
I(\eta,\tau;1)=\frac{\tau^4}{1+\eta}\arctan\left
  (\frac{1}{1+\eta}\right )+\mathcal{O}(\tau^6).$$
In this case we see that $C_{\alpha,2}(0_+)=0$ and
$C_{\alpha,4}(0_+)=\pi/4$, and they provide a good approximation for the 'true'
value of $I(0_+,\tau;1)$. 
\end{remark}

\begin{remark}\label{remarca23} The contribution 
$I_{\alpha,{\rm cot}}:=\tau^4C_{\alpha,4}(0_+)$ is the so-called
cotunneling current. Further discussion on it will be given in Section 5.2.
Here we only stress that in the absence of the bias, $I_{\alpha,{\rm
    cot}}=0$ because in this case the chemical potentials of the 
leads are equal and hence ${\cal P}_{\gamma\alpha}={\cal P}_{\alpha\gamma}$.  In the off-resonant non-interacting case, we can prove that 
it does not depend on the initial state in the sample and it is given by the first term of the Landauer formula (see \eqref{cotLB}). 
\end{remark}

\begin{remark}\label{remarca244}{\bf Memory effects and dependence on the initial state}. 
For small samples one is able to simplify the formula giving  
the cotunneling current. A typical example is a two-site quantum
dot. Let us denote by 
$e_{1,2}$ the eigenvalues of the non-interacting dot. We also have
that $h_S\phi_1=e_1\phi_1$ and $h_S\phi_2=e_2\phi_2$. The four 
many-body states are $E_1=0$ (empty sample), $E_2=e_1$ (the ground state
of $h_S$), $E_3=e_2$ (the excited
state of $h_S$) and $E_4=e_1+e_2+U$ (fully occupied)
where $U$ denotes the strenght of the Coulomb interaction. Let us
consider that the initial state of the system
is $|\nu_0\rangle= |10\rangle$ and $E_{\nu_0}=e_1$, which means that
before the coupling we start with exactly one electron in the sample, occupying
the lowest level. 

The two spectral conditions imposed by the off-resonant regime have to be
checked for any given set of parameters. Let us
explicitely write down these conditions for a sample having only two sites: 
a). $e_{1,2}\notin [-2t_L,2t_L]$, $e_{1,2}+U\notin [-2t_L,2t_L]$; 
b). $E_4-E_1=e_1+e_2+U\notin [-4t_L,4t_L]$.
These conditions can be satisfied in many situations, for example when
both $e_{1,2}$ are either very negative or very positive such that they are
far away from the spectrum of the leads.

The cotunneling current in Eq.(\ref{J4fin}) can be further 
simplified by calculating the coefficients $A$ and $A^*$.
In order to do that we have to express the creation and annihilation
operators in the contact sites 
$a^\#(|m_{\alpha}\rangle) $ and $a^\#(|m_{\gamma}\rangle) $ in terms of creation and annihilation
operators in a given single-particle eigenstates 
$a^\#(|\phi_1\rangle) $ and $a^\#(|\phi_2\rangle)$. This leads to obvious selection rules for the 
MBS. Calculating the cotunneling rates terms one can identify
elastic and inelastic contributions to the current:
\begin{equation}
I_{\alpha, {\rm cot}}=\tau^4C_{\alpha,4}(0_+)=I_{{\rm el}}+I_{{\rm in}},
\end{equation}
where:

\begin{align}\label{Jelastic}
I_{{\rm el}}&=\frac{\tau^4}{\pi^2t_L^2}\sum_{\gamma} \int_{-2t_L}^{2t_L}dE\;\left (1-\frac{E^2}{4t_L^2}\right )\left |
  \frac{\phi_1(m_{\alpha})\overline{\phi_1(m_{\gamma})}}{E-e_1}
+ \frac{\phi_2(m_{\alpha})\overline{\phi_2(m_{\gamma})}}{E-e_2-U}\right
|^2(f_{\alpha}(E)-f_{\gamma} (E))
\end{align}
and 
\begin{align}\nonumber
&I_{{\rm in}}=\frac{\tau^4}{\pi^2t_L^2}\int_{-2t_L}^{2t_L} dE\;\left (1-\frac{E^2}{4t_L^2}\right )\\
&\nonumber \left \lbrace \chi_L(E+e_1-e_2)|\phi_1(m_{\gamma})|^2|\phi_2(m_{\alpha})|^2
\frac{f_{\alpha}(E)(1-f_{\gamma} (E+e_1-e_2))}{e_2+U-E} 
\left (\frac{1}{e_2-E}-\frac{1}{e_2+U-E}\right) \right . \\\nonumber
&+\left. \chi_L(E+e_2-e_1)|\phi_1(m_{\gamma})|^2|\phi_2(m_{\alpha})|^2
  \frac{f_{\alpha}(E+e_2-e_1)
(1-f_{\gamma} (E))}{e_1-E}
\left (  \frac{1}{e_1+U-E}-\frac{1}{e_1-E}  \right ) \right\rbrace  \\  \label{Jinel}
&-\lbrace \alpha\leftrightarrow \gamma\rbrace .
\end{align}

Let us comment on the two contributions to the cotunneling in this case. 
Obviously $I_{{\rm el}}$ is given by a Landauer formula, even if the interaction strength $U$
appears in one of the denominators. The two electrons implied in the pairwise tunneling have the same energy $E$
hence this is elastic cotunneling. 
This contribution can be compared with the one calculated in Ref. \cite{Pedersen} via what the authors call the 'T-matrix method'.        
To make the connection to their results one should use the cotunneling rate $\gamma^{RL}_{11}$ given in Eq. (19) of 
Ref. \cite{Pedersen} and calculate the steady-state current as $\gamma^{RL}_{11}-\gamma^{LR}_{11}$. 

In contrast, $I_{{\rm in}}$ can no longer be written in a Landauer form and contains 
inelastic processes, as the energies in the Fermi functions do not
coincide. Note that $I_{{\rm in}}$ 
vanishes in the non-interacting case: this happens because for $U=0$
the contributions of various inelastic processes cancel each
other. Moreover, if $|e_2-e_1|>4t_L$ then $\chi_L(E+e_2-e_1)$ and
$\chi_L(E+e_1-e_2)$ will vanish for all $E\in
[-2t_L,2t_L]$, thus again $I_{{\rm in}}=0$. But otherwise it is
nonzero. 

We can repeat this computation choosing the initial condition
$|\nu_0\rangle=|00\rangle $ (the sample is empty before coupling it to
the leads). In this case we find:
\begin{equation}\label{elastycu}
I_{\alpha, {\rm cot}}=\frac{\tau^4}{\pi^2t_L^2}\sum_{\gamma} \int_{-2t_L}^{2t_L}dE\;\left (1-\frac{E^2}{4t_L^2}\right )\left |
  \frac{\phi_1(m_{\alpha})\overline{\phi_1(m_{\gamma})}}{E-e_1}
+ \frac{\phi_2(m_{\alpha})\overline{\phi_2(m_{\gamma})}}{E-e_2}\right
|^2(f_{\alpha}(E)-f_{\gamma} (E)).
\end{equation}
 Otherwise stated, for this initial state of the sample 
the cotunneling current is given by the non-interacting
Landauer formula.  This means that the cotunneling current in the
  interacting case {\it depends on the 
initial conditions of the sample}. This is not such an unexpected
result, as different initial many-body configurations 
of the sample select different relevant cotunneling processes. 
We stress though that this memory effect concerns only 
the cotunneling current in the off-resonant regime. In the resonant
case where sequential and cotunneling processes coexist
we do not expect this to happen.   
\end{remark}

\section{Proof of {\rm (i)}: thermodynamic limit and the definition of the transient}\label{limitatermo}

In mesoscopic quantum transport we have to deal with two 
aparently contradictory conditions: 1). the leads must be finite if we want the 
total density matrix to be trace class, and in that case the total Hamiltonian has
purely discrete spectrum; 2). the total Hamiltonian must also have some 
continuous spectrum since otherwise the ergodic current would be identically
zero. The correct way out is to fix the time $t$, define the expectations at finite
leads and afterwards make them infinitely long. Only {\it after} the
thermodynamic limit we can let $t$ to go to infinity. More than that, the total density matrix 
is not the good object to work with, and any formal perturbative
expansions in $\tau$ at $t=\infty$ {\it before} the thermodynamic
limit has no clear mathematical meaning.  

In this section unless otherwise stated the leads are assumed to be of finite length $\Lambda$. But for the simplicity
of writing we omit the label $\Lambda$ on the leads' Hamiltonian. In order to get an expansion of the 
current in powers of the tunneling Hamiltonian we define $W(t)=e^{itH_0}e^{-itH}$, verifying the equation: 
\begin{equation}\label{doubleW}
i\dot W(t)={\tilde H}_T(t)W(t), \qquad W(0)=1, \qquad {\tilde H}_T(t):=e^{itH_0}H_Te^{-itH_0}. 
\end{equation}
Then the solution is:
\begin{eqnarray}\nonumber
W(t)&=&1-i\int_{t_0}^t ds {\tilde H}_T(s)W(s)\\\label{WS} 
&=&1+\sum_{k\geq 1}(-i)^k\int_{t_0}^t ds_1\int_{t_0}^{s_1}ds_2...
\int_{t_0}^{s_{k-1}}ds_k{\tilde H}_T(s_1){\tilde H}_T(s_2)...{\tilde H}_T(s_k).
\end{eqnarray} 
Using the ciclicity of the trace and the definition of $W(t)$ one rewrites Eq.(\ref{currTDL}) as follows:
\begin{equation}\label{JTD}
\langle J_{\alpha}(t) \rangle_{{\rm ref}} =\lim_{\Lambda\to\infty}{\rm Tr}_{{\cal F}}
\{\rho_0^{(\Lambda)} W^*(t){\tilde J}_{\alpha}(t)W(t) \},\quad {\tilde J}_\alpha(t):=e^{itH_0}J_\alpha e^{-itH_0}.
\end{equation}
It is clear that by replacing $W(t)$ as given by Eq.(\ref{WS}) in Eq.(\ref{JTD}) one obtains a full expansion 
of the current w.r.t the tunneling Hamiltonian $H_T$. Our strategy is to show that one can perform the 
thermodynamic limit on each term in this expansion. Let us make a few remarks on the structure of these terms 
and give the main steps we follow for calculating them. 

i) Given the structure of $W(t)$ and $H_T$ the current will be a series of monomials containing combinations 
of creation/annihilation operators from both the leads and the sample. However, due to the particular tensor product 
form of $\rho_0^{(\Lambda)}$, the particle number conservation 
requires that in all monomials with a non-vanishing contribution to the trace, the number of creation operators should equal the number 
of annihilation operators separately for the sample, and for each lead. It also means that each such monomial 
contains an odd number of $H_T$'s and is of even order in $\tau$ since the current operator itself is proportional with $\tau$.      

ii) In order to simlify notation, we write $a^{\#}(x)$ instead of $a^{\#}(|x\rangle)$, that is we identify the site $x$ with the basis vector 
$|x\rangle$. We deal with the operators acting on ${\cal F}_S$ by systematically inserting the projections of 
many-body states $\{|\nu \rangle\langle \nu| \}$ between any two $H_T$'s. 
Using the matrix elements $A_{\nu\nu'}$ introduced above (see Eq.(\ref{Anu})) and the shorthand notation 
${\tilde a}^{\#}_t(x)=e^{itH_0}a^{\#}(x)e^{-itH_0} $ one has for example:
\begin{equation}
\langle \nu_0,{\tilde H}_T(s_k) \nu_k\rangle =\sum_{\alpha_k}e^{is_k(E_{\nu_0}-E{\nu_k})}
\left [A_{\nu_0\nu_k}(m_{\alpha_k}){\tilde a}^*_{s_k}(0_{\alpha_k})
+A^*_{\nu_0\nu_k}{\tilde a}_{s_k}(0_{\alpha_k}) \right ]. 
\end{equation} 
Note that $A_{\nu\nu'}$ couples many-body states whose particle number differ by at most one. Also, 
$A_{\nu\nu'}$ does not depend on $\Lambda$, thus the thermodynamic
limit is only relevant for terms of the type: 
$${\rm Tr}_{{\cal F}_L}\{ \rho^{(\Lambda)}_L{\tilde a}^{\#_1}_{s_1}(0_{\alpha_1})...{\tilde a}^{\#_{2N}}_{s_{2N}}
(0_{\alpha_{2N}})  \}.$$

iii) Next we change the representation of the operators using the eigenstates $\varphi_{q_{\alpha_k}}$ of the leads' Hamiltonian:
\begin{equation}\label{aq}
{\tilde a}^{\#_k}_{s_k}(0_{\alpha_k})=\sum_{q_{\alpha_k}}e^{is_k\theta_{\#_k}\varepsilon(q_{\alpha_k})}
\varphi^{\#_k}_{q_{\alpha_k}}(0_{\alpha_k})a^{\#_k}_{q_{\alpha_k}},  
\end{equation}
where we introduced the notations:
\begin{eqnarray}\label{cor1}
\theta_{\#_k}=\left\lbrace \begin{array}{ccc}
+  \quad {\rm for} \quad a^*_{q_{\alpha_k}}, \\
-\quad {\rm for}  \quad  a_{q_{\alpha_k}}.
\end{array}\right.\quad
\varphi^{\#_k}_{q_{\alpha_k}}(0_{\alpha_k})=\left\lbrace \begin{array}{ccc} 
\overline {\varphi (0_{\alpha_k})} \qquad {\rm for}  \quad a^*_{q_{\alpha_k}}, \\
\varphi (0_{\alpha_k}) \qquad {\rm for}  \quad a_{q_{\alpha_k}},
\end{array}\right.
\end{eqnarray}
and $\varphi (0_{\alpha})=\langle \varphi_{q_{\alpha}},0_{\alpha} \rangle$. 
The general term on which one should perform the thermodynamic limit reads as follows: 
\begin{equation}\label{2Nterm}
\sum_{\vec\alpha}\sum_{ q_{\vec\alpha}}\sum_{ \vec\# }e^{i\theta_{\#_1}s_1
\varepsilon(q_{\alpha_1})+...+i\theta_{\#_{2N}}s_{2N}\varepsilon(q_{\alpha_{2N}})}
\varphi^{\#_1}_{q_{\alpha_1}} (0_{\alpha_1})..\varphi^{\#_{2N}}_{q_{\alpha_k}}   (0_{\alpha_{2N}}) {\rm Tr}_{{\cal F}_L}
\{\rho_L^{(\Lambda)} a^{\#_1}_{q_{\alpha_1}}..a^{\#_{2N}}_{q_{\alpha_{2N}}} \},
\end{equation}
where in the trace above there are precisely $N$ creation and $N$
annihilation operators from the leads. We introduced the 
shorthand notations $\vec\alpha := (\alpha_1,..,\alpha_{2N})$, 
$ q_{\vec\alpha}:= q_{\alpha_1},..,q_{\alpha_{2N}}$ and $ \vec\#:=\#_1,..\#_{2N}$ . The trace is further
calculated using the Wick theorem (see \cite{Fetter}) which holds because the leads are noninteracting. 
The idea behind the Wick procedure is to systematically use the anticommutation relations in order to 
reduce the monomial of order $2N$ to a sum of monomials of order $2N-2$. The simplest case corresponds 
to all six combinations for $N=2$. For example:
\begin{eqnarray}\nonumber
{\rm Tr}_{{\cal F_L}}\{\rho_L^{(\Lambda)} a^*_{q_{\alpha_1}}a^*_{q_{\alpha_2}}a_{q_{\alpha_3}}a_{q_{\alpha_4}}\}
&=&-\delta_{q_{\alpha_1}q_{\alpha_3}}\delta_{q_{\alpha_2}q_{\alpha_4}}f_{\alpha_1}(\varepsilon_{q_{\alpha_1}})
f_{\alpha_2}(\varepsilon_{q_{\alpha_2}})\\\label{Wick1} 
&+&\delta_{q_{\alpha_1}q_{\alpha_4}}\delta_{q_{\alpha_2}q_{\alpha_3}}f_{\alpha_1}(\varepsilon_{q_{\alpha_1}})
f_{\alpha_2}(\varepsilon_{q_{\alpha_2}}),
\end{eqnarray} 
where we used the cyclicity of the trace, the identity
$a^*_{q_{\alpha_1}}\rho_L^{(\Lambda)}=e^{\beta \varepsilon_{q_{\alpha_1}}}\rho_La^*_{q_{\alpha_1}} $
and the well known fact 
${\rm Tr}_{{\cal F}_L}\{\rho_L^{(\Lambda)} a^*_{q_{\alpha}}a_{q_{\beta}}\}=\delta_{\alpha\beta}
f_{\alpha}(\varepsilon_{q_{\alpha}}) $, where $f_{\alpha}$ is the Fermi function associated to lead $\alpha$.
One can easily show that all allowed combinations of 4 operators can be expressed in terms of 
products $ff$, $\overline f f$ and $\overline f \;  \overline f$, where $\overline f =1-f$.
Also, it is important to observe that due to the Kronecker symbols the sums over $q$'s are reduced and
one actually obtains products of terms which are of the following type:
\begin{align}\label{aprilie17}
\sum_{q_{\beta}}e^{\pm i(s-s')\varepsilon_{q_{\beta}}} f_{\beta}(\varepsilon_{q_{\beta}}) \langle 0_{\beta},\varphi_{q_{\beta}}\rangle
\langle \varphi_{q_{\beta}},0_{\beta} \rangle &=\langle 0_{\beta}, e^{\pm i(s-s')h_{\beta}^{(\Lambda)}}f_{\beta}(h_{\beta}^{(\Lambda)})0_{\beta} \rangle , \\
\sum_{q_{\beta}}e^{\pm i(s-s')\varepsilon_{q_{\beta}}} \langle 0_{\beta},\varphi_{q_{\beta}}\rangle
\langle \varphi_{q_{\beta}},0_{\beta} \rangle &=\langle 0_{\beta}, e^{\pm i(s-s')h_{\beta}^{(\Lambda)}}0_{\beta} \rangle .\nonumber
\end{align}
The second term appears from combinations containing $\overline f \;  \overline f$. 

For terms of higher order one proceeds in a similar way using the general formula (see Eq.(24.36)) in Ref.\cite{Fetter}:
\begin{align}\nonumber
{\rm Tr}_{{\cal F}_L}\{\rho_L^{(\Lambda)} a^{\#_1}_{q_{\alpha_1}}..a^{\#_{2N}}_{q_{\alpha_{2N}}} \}&=
\{a^{\#_1}_{q_{\alpha_1}},a^{\#_2}_{q_{\alpha_2}} \}_+f^{\#_1}(\varepsilon_{q_{\alpha_1}})
{\rm Tr}_{{\cal F}_L}\{\rho_L^{(\Lambda)}a^{\#_3}_{q_{\alpha_3}} a^{\#_4}_{q_{\alpha_4}}..a^{\#_{2N}}_{q_{\alpha_{2N}}} \} \\\nonumber
&-\{a^{\#_1}_{q_{\alpha_1}},a^{\#_3}_{q_{\alpha_3}} \}_+f^{\#_1}(\varepsilon_{q_{\alpha_1}})
{\rm Tr}_{{\cal F}_L}\{\rho_L^{(\Lambda)}a^{\#_2}_{q_{\alpha_2}}a^{\#_4}_{q_{\alpha_4}} ..a^{\#_{2N}}_{q_{\alpha_{2N}}} \}+... \\\label{Wick2}
&+\{a^{\#_1}_{q_{\alpha_1}},a^{\#_{2N}}_{q_{\alpha_{2N}}} \}_+f^{\#_1}(\varepsilon_{q_{\alpha_1}})
{\rm Tr}_{{\cal F}_L}\{\rho_L^{(\Lambda)} a^{\#_2}_{q_{\alpha_2}}a^{\#_3}_{q_{\alpha_3}}..a^{\#_{2N-1}}_{q_{\alpha_{2N-1}}} \}.
\end{align}

Thus we have shown that the thermodynamic limit is to be performed
only on factors
like in \eqref{aprilie17}. We give this result as a general lemma:
\begin{lemma}\label{termoman}
Let $\mathcal{N}_\Lambda$ be the set $\{0,1,\dots,\Lambda\}$ with $\Lambda\leq
\infty$. Let $h_\infty$ be the discrete Laplace operator on the
halfline $\mathcal{N}_\infty$
with Dirichlet boundary condition at $-1$, and $h_\Lambda$ is the
restriction of $h_\infty$ on $\mathcal{N}_\Lambda$ with Dirichlet
conditions at $-1$ and $\Lambda+1$. Let $F$ be any continuous function
defined on the interval $[-2t_L,2t_L]$. Then we have:
\begin{equation}\label{aprilie171}
\lim_{\Lambda\to \infty} \langle 0,F(h_\Lambda)0\rangle =
\langle 0,F(h_\infty)0\rangle =\frac{1}{\pi
  t_L}\int_{-2t_L}^{2t_L}\sqrt{1-\frac{E^2}{4t_L^2}}\; F(E)\; dE.
\end{equation}
\end{lemma}
\begin{proof} Fix 
  $\epsilon>0$. The spectrum of all $h_\Lambda$'s is contained in
  $[-2t_L,2t_L]$. Because $F$ is continuous on this interval, it can
  be uniformly approximated with polynomials. The Weierstrass
  approximation theorem says that there exists a polynomial
  $P_\epsilon(x)=\sum_{j=0}^{N}a_jx^j$ such that  
\begin{equation}\label{aprilie172}
||F-P_\epsilon||_\infty:=\sup_{x\in [-2t_L,2t_L]}|F(x)-P_\epsilon(x)|\leq \epsilon/3.
\end{equation}
The spectral theorem implies that 
$||F(A)-P_\epsilon(A)||=||F-P_\epsilon||_\infty$ for any self-adjoint
operator $A$ whose spectrum lies in $[-2t_L,2t_L]$. Thus we can write:
\begin{equation}\label{aprilie173}
|\langle 0,F(h_\infty)0\rangle -\langle
0,P_\epsilon(h_\infty)0\rangle|\leq \epsilon/3,\quad |\langle
0,F(h_\Lambda)0\rangle -\langle
0,P_\epsilon(h_\Lambda)0\rangle|\leq \epsilon/3,\quad \forall \Lambda
\geq 1.
\end{equation}
It is very important to note that the above estimate holds true
uniformly in $\Lambda$. Now let us remark that there exists
$\Lambda_\epsilon$ sufficiently large such that 
\begin{equation}\label{aprilie174}
\langle
0,P_\epsilon(h_\infty)0\rangle =\langle
0,P_\epsilon(h_\Lambda)0\rangle,\quad \forall \Lambda\geq \Lambda_\epsilon.
\end{equation}
The explanation is that $h_\Lambda^k|0\rangle=h_\infty^k|0\rangle$ if
$k\leq \Lambda$, because we cannot reach the 'other' boundary after
less than $\Lambda$ steps. Thus choosing $\Lambda_\epsilon$ larger
than the degree of $P_\epsilon$ is sufficient to conclude that 
$P_\epsilon(h_\Lambda)|0\rangle=P_\epsilon(h_\infty)|0\rangle$. Now
using \eqref{aprilie173} and \eqref{aprilie174} we have:
$$|\langle 0,F(h_\Lambda)0\rangle -
\langle 0,F(h_\infty)0\rangle|\leq 2\epsilon/3<\epsilon,\quad \forall
\Lambda\geq \Lambda_\epsilon, $$
and the proof is over.
\end{proof}

\vspace{0.5cm}

After applying the Lebesgue dominated convergence theorem on the iterated
integrals of \eqref{JTD}, we arrive after some work 
at a rough estimate of the form: 
\begin{equation}\label{aprilie175}
|I_{\alpha,t}(\tau)|\leq \sum_{n\geq 1}\tau^{2n}C^{2n}\frac{t^{2n}}{(2n)!},
\end{equation}
where $C$ is some positive constant. Thus $I_{\alpha,t}(\cdot)$ is entire in $\tau$. But this estimate only
says that the transient current cannot grow faster than an exponential of the type
$\tau^2 e^{C\tau t}$, which is not very useful if $t$ is large. But at least if $\eta$ is
chosen such that $\tau/\eta$ is small enough, then \eqref{Jeta} holds
true.

\section{Off-resonant transport}\label{offresonant}

Before starting our calculations we review the Landauer formula for non-interacting electrons \cite{CJM1} 
which was proved to give the steady-state current both for discrete and continuous models at arbitrary bias 
\cite{Aschb, Nenciu,CDNP-1}.
The reason to make some connection between the Landauer formula and our results is twofold. On one hand any 
calculation in the interacting case should lead to this formula when the interaction strength $U$ is set back to zero. 
On the other hand one can get some general facts about the expansion of current in powers of the lead-dot tunneling $\tau$.
The Landauer formula gives the steady-state current in the lead $\alpha$:
\begin{equation}\label{LB}
I_{\alpha,\infty}(\tau)=\sum_{\gamma}\int_{-2t_L}^{2t_L} dE (f_{\alpha}(E)-f_{\gamma}(E))|{\cal T}_{\alpha\gamma}(E)|^2,
\end{equation}
where the transmittance ${\cal T}_{\alpha\gamma}(E)$ is defined as follows (see \cite{CJM1}):
\begin{equation}\label{calT}
{\cal T}_{\alpha\gamma}(E)=\frac{\tau^2}{\pi t_L}
  \sqrt{1-\frac{E^2}{4t_L^2}}\left\langle m_{\alpha},\left
  (h_S-E-\frac{\tau^2}{t_L}\zeta_1^+(E)\Pi_T\right )^{-1}m_{\gamma}\right\rangle,
\end{equation}
where we introduced the orthogonal projection on the contact sites $m_{\beta}$,
$\Pi_T:=\sum_{\beta} |m_{\beta}\rangle\langle m_{\beta}| $ and
$\zeta_1(z)={\zeta}_+(z)$ if ${\rm Im}(z)>0$, $\zeta_1(z)={\zeta}_-(z)$ if ${\rm Im}(z)<0$, where:
\begin{equation}\label{matqa3}
\zeta_\pm(z)=\frac{z}{2t_L} \mp i\sqrt{1-z^2/(4t_L^2)},\quad z\not\in
((-\infty,-2t_L]\cup [2t_L,\infty)). 
\end{equation}
If all eigenvalues 
$e_{\lambda}$ of $h_S$ are far away from the 
spectrum of the leads, then the ergodic current becomes analytic near
$\tau=0$ and the leading term is of order $\tau^4$ and coincides with
\eqref{elastycu}. 

In contrast, if some eigenvalue
$e_{\lambda}$ of the sample is inside $(-2t_L,2t_L)$, the ergodic current has a completely different behavior with $\tau$. For simplicity, assume that 
 all other eigenvalues of $h_S$ are outside $[-2t_L,2t_L]$, while $e_{\lambda}$ is non-degenerate and corresponds to an eigenvector $\phi$, i.e. 
$h_S\phi=e_{\lambda}\phi$. Then following \cite{CJM1} one can prove:
\begin{align}\label{LB33}
&I_{\alpha,\infty}(\tau)\nonumber \\
&=\frac{\tau^4}{\pi^2 t_L^2}\sum_{\gamma}\int_{-2t_L}^{2t_L} dE 
|\phi(m_{\gamma})|^2|\phi(m_{\alpha})|^2(f_{\alpha}(E)-f_{\gamma}(E))\frac{1-\frac{E^2}{4t_L^2}}{|E-e_\lambda 
+\frac{\tau^2}{t_L}\zeta_1^+(E)\langle \phi,\Pi_T\phi\rangle|^2 }+\mathcal{O}(\tau^4)\nonumber \\
&=\tau^2\left (C(e_\lambda,t_L) \sum_\gamma [f_{\alpha}(e_\lambda)-f_{\gamma}(e_\lambda)]\frac{ |\phi(m_{\gamma})|^2|\phi(m_{\alpha})|^2}{  
\sum_\beta |\phi(m_{\beta})|^2}   +\mathcal{O}(1) \right ),
\end{align}
where $C(e_\lambda,t_L)$ is some constant. It is clear that this expression has nothing in common with \eqref{seq}, which only contains $f_\alpha$ and not 
differences of Fermi functions. 

\subsection{Proof of {\rm (ii)}: Sequential tunneling contribution}

In this Section we calculate the first two contributions to the steady-state current, that is the terms of 
order two and four in the transfer Hamiltonian. Using the identity:
\begin{equation}
e^{-itH}e^{itH_0}=1-i\int_0^t ds e^{-isH}H_Te^{isH_0}
\end{equation}
and expanding the unitary evolution $e^{-itH} $ up to the 2nd order in $H_T$ one gets from
\eqref{JTD} and (\ref{Jeta}):
\begin{equation}
C_{\alpha,2}(\eta)= i\eta\int_0^{\infty}dt\; e^{-\eta t}\int_0^t ds 
\langle [\tilde H_T(-s),J_{\alpha}] \rangle_{{\rm ref}}.
\end{equation}
By replacing $J_{\alpha}$ and $H_T$ one arrives after straightforward calculation at the following expression:
\begin{eqnarray}\nonumber
{\rm Tr}_{\cal F}\{\rho_0^{(\Lambda)} [\tilde H_T(-s),J_{\alpha}]\} &=&\tau^2\sum_{\nu} 
\sum_{q_{\alpha}} |\varphi_{q_{\alpha}}(0_{\alpha})|^2 
|A_{\nu\nu_0 }|^2(1-f_{\alpha}(\varepsilon_{q_{\alpha}})
( e^{is(E_{\nu}-E_{\nu_0}+\varepsilon_{q_{\alpha}})} + c.c ) \\\nonumber
&-&\tau^2\sum_{\nu}\sum_{q_{\alpha}} |\varphi_{q_{\alpha}}(0_{\alpha})|^2 |A^*_{\nu\nu_0 }|^2
f_{\alpha}(\varepsilon_{q_{\alpha}})
( e^{is(E_{\nu}-E_{\nu_0}-\varepsilon_{q_{\alpha}})} + c.c ) \\\nonumber 
&=&\tau^2\sum_{\nu} |A_{\nu\nu_0 }|^2 (e^{is(E_{\nu}-E_{\nu_0})}
\langle 0_{\alpha},e^{ish_L}(1-f (h_L))0_{\alpha}\rangle + c.c  ) \\\label{Jseq}
&-&\tau^2\sum_{\nu}|A^*_{\nu\nu_0 }|^2 (e^{is(E_{\nu}-E_{\nu_0})}\langle 0_{\alpha},e^{-ish_L}f (h_L)0_{\alpha}\rangle  + c.c).
\end{eqnarray}
In the thermodynamic limit one has:
\begin{equation}
\langle 0_{\alpha},e^{-ish_L}f (h_L)0_{\alpha}\rangle =\int_{-2t_L}^{2t_L}dE |\varphi^{\alpha}_E(0_{\alpha})|^2
e^{-isE}f_{\alpha}(E),
\end{equation}
where $\varphi^{\alpha}_E$ denotes the generalized eigenfunction of the semiinfinite lead corresponding to energy $E$ (see \eqref{semiinf}).  
By performing the time integral over $s$ the contribution of order $\tau^2$ to the transient current is obtained as:
\begin{equation}\label{J2trans}
 2\tau^2 \sum_{\nu} \int_{-2t_L}^{2t_L}dE |\varphi^{\alpha}_{E}(0_{\alpha})|^2
\left ( |A_{\nu\nu_0 }|^2(1-f_{\alpha})\frac{\sin(\Delta^+_{\nu\nu_0}t)}{\Delta^+_{\nu\nu_0}}
-f_{\alpha}|A^*_{\nu\nu_0 }|^2\frac{\sin(\Delta^-_{\nu\nu_0}t)}{\Delta^+_{\nu\nu_0}}\right),
\end{equation} 
where for simplicity we omitted to write the energy dependence of the Fermi functions and we introduced the notations: 
\begin{equation}
\Delta^{\pm}_{\nu\nu_0}(E):=E_{\nu}-E_{\nu_0}\pm E
\end{equation}
Then we perform the final time integral and use the identity:
\begin{equation}\label{eta}
\lim_{\eta\to 0}\eta\int_0^{\infty}dt e^{-\eta t} \frac{e^{it\Delta^{\pm}_{\nu\nu_0}}-e^{-it\Delta^{\pm}_{\nu\nu_0}}}
{\Delta^{\pm}_{\nu\nu_0}}=2\pi\delta (\Delta^{\pm}_{\nu\nu_0}).
\end{equation}
to arrive at Eq.(\ref{seq}).

\subsection{Proof of {\rm (iii)}: Cotunneling}
The 4th order contribution to the current follows from the expansion of the unitary evolution up to the 3rd order in 
the tunneling operator $H_T$: 
\begin{align}\nonumber
&\tau^4C_{\alpha,4}(\eta)=i\int_0^tds_1\int_0^{s_1}ds_2\int_0^tds'_1\langle e^{-is_1H_0}H_Te^{is_2H_0}H_Te^{-is_2H_0}e^{is_1H_0}J_{\alpha}
e^{-is'_1H_0}H_Te^{is'_1H_0}\rangle_{{\rm ref}}
\\\label{J4}
&-i\int_0^tds_1\int_0^{s_1}ds_2\int_0^{s_2}ds_3\langle e^{-is_1H_0}H_Te^{is_3H_0}H_Te^{-is_3H_0}e^{is_2H_0}H_T
e^{-is_2H_0}e^{is_1H_0}J_{\alpha}\rangle_{{\rm ref}}+ c.c.
\end{align}
In order to achieve a more explicit form of $C_{\alpha,4}(\eta)$ we follow the same steps as in the 
proof of the thermodynamic limit, that is we insert the many-body states of $H_S$ in order to deal with the 
operators acting on ${\cal F}_S$, then we switch to the proper basis of $H_L$  and finally use the Wick theorem for 
all non-vanishing combinations of the type ${\rm Tr} \{a^{\#}_{q_1}a^{\#}_{q_2}a^{\#}_{q_3}a^{\#}_{q_4}\}$. The
calculations are tedious but straightforward. We find that there are 48 terms contributing to the cotunneling current. 
At the next step we perform the time integrals.  It is sufficient to calculate the real part of this integrals because 24 terms are the complex 
conjugates of the remaining ones.  Moreover, one notes that there are only two types of
integrals:
\begin{eqnarray}\nonumber 
B_1(t)&=&\int_0^tds_1\int_0^{s_1}ds_2\int_0^tds_3 \cos(s_1x+s_2y+s_3z)\\\nonumber
&=&\frac{1}{zy(x+y)}( \sin tz -\sin t(x+y+z)+\sin t(x+y)) \\\label{I1}
&+&\frac{1}{xyz}(\sin t(x+z)- \sin tz + \sin tx),\\\nonumber
B_2(t)&=&
\int_0^tds_1\int_0^{s_1}ds_2\int_0^{s_2}ds_3 \cos(s_1x+s_2y+s_3z)\\\label{I2}
&=&\frac{\sin t(x+y)}{zy(x+y)}+\frac{\sin tx}{zx(z+y)}-\frac{\sin tx}{xyz}-\frac{\sin t(x+y+z)}{z(y+z)(x+y+z)},
\end{eqnarray}
 where $x,y,z$ contain two many-body energies of $H_S$ and energy of one or two electrons from the leads
(an example is $x=E_{\nu''}-E_{\nu_0}+\varepsilon_{q_1}$, $y=E_{\nu'}-E_{\nu''}-\varepsilon_{q_2}$, 
$z=E_{\nu}-E_{\nu'}+\varepsilon_{q_2} $).
Then one has to perform the thermodynamic limit, to calculate the integral over time and take 
the limit $\eta\to 0$. This final step brings in plenty of delta functions. Our first off-resonant condition was that 
$E_{\nu}-E_{\nu'}-\varepsilon_{q}\neq 0$ if the number of electrons in the 
MBS $|\nu\rangle,|\nu'\rangle$ differ by one.  Our second off-resonant condition implies that 
$E_{\nu}-E_{\nu'}\pm (\varepsilon_{q_1}+\varepsilon_{q_2})\neq 0$ , for any pair of many-body energies 
$E_{\nu},E_{\nu}$ whose particle numbers differ by two. By analyzing all combinations of $x,y,z$ 
it follows that the remaining off-resonant terms arise from $\delta(x)/zy$ for $B_1$ and from $\delta(x+y)/zy$ for $B_2$.
In these terms the delta functions impose conditions of the form 
$E_{\nu'}-E_{\nu_0}+\varepsilon_{q_1}-\varepsilon_{q_2}=0$, which means that the dot initially in the state 
$|\nu_0 \rangle $ passes to the state $|\nu ' \rangle $ by exchanging {\it two} electrons with the leads. 
This process is called {\it cotunneling} in the physical literature, 
because the electrons now tunnel pairwise. 
After collecting all these terms and taking advantage of some cancelations one arrives at
the final expression for the cotunneling current given by Eq.(\ref{cot1}) of the theorem.

Let us make a few remarks on the cotunneling current. From the sequence of $A$'s appearing Eq.(\ref{J4fin}) 
one observes that the cotunneling processes always imply different leads.
Take for example the 3rd term. It describes the following sequence: an electron with energy $E$ enters the dot 
from the lead $\alpha$, while the second electron of energy $E'=E-E_{\nu '}+E_{\nu_0}$ leaves the 
dot to lead $\gamma$. The remaining two terms described the reverse process: the electron tunnels back from 
the lead $\gamma$ and the second one tunnels out to lead $\alpha$. The other terms can be described in a similar way.  
Also note that the cotunneling contributions explicitely contain the initial state of sample $\nu_0$. 

A natural question is what we can say about the cotunneling current in the non-interacting case.
Let us recall here that $e_{\lambda}$ are the eigenvalues of $h_S$, i.e $h_S\phi_{\lambda}=e_{\lambda}\phi_{\lambda}$.
Then the operators $a^{\#}(|m_{\alpha}\rangle)$ and $a^{\#}(|m_{\gamma}\rangle)$ appearing in the coefficients $A$ 
in Eq.(\ref{cot1}) can be written in terms of $a^{\#}_\lambda:=a^{\#}(|\phi_\lambda\rangle)$. 
Moreover, the sums over the many-body states of $H_S$ allow one to recover the resolvent $(H_S-E_{\nu_0}-\varepsilon)^{-1}$
and also the Fermi-Dirac operator $f_{\alpha,\gamma}(H_S-E_{\nu_0}-\varepsilon)$. 
As an example we consider the 2nd term in Eq.(\ref{J4fin}. Introducing the notation 
${\tilde f}_{\gamma}(E):=\chi_L(E)f_{\gamma}(E)$ one has:
\begin{align}\nonumber
M_2&:=-\sum_{\nu,\nu',\nu''}  \frac{(1-f_{\alpha}(E)) {\tilde f}_{\gamma}(E+E_{\nu'}-E_{\nu_0} )}
{(E_{\nu}-E_{\nu_0}+E) (E_{\nu_0}-E_{\nu''}-E)}
A_{\nu_0\nu}^*(m_{\alpha})A_{\nu\nu'}(m_{\gamma})A_{\nu'\nu''}^*(m_{\gamma})A_{\nu''\nu_0}(m_{\alpha})\\\nonumber
&=\sum_{\lambda_1,\lambda_2,\lambda_3,\lambda_4}\langle \phi_{\lambda_1},m_{\alpha} \rangle 
\langle m_{\gamma}, \phi_{\lambda_2} \rangle 
\langle \phi_{\lambda_3},m_{\gamma} \rangle \langle m_{\alpha},\phi_{\lambda_4} \rangle (1-f_{\alpha}(E))
{\tilde f}_{\gamma}(E+E_{\nu'}-E_{\nu_0} ) \\\nonumber
&\times \langle \nu_0,a^*_{\lambda_1}(H_S-E_{\nu_0}-E)^{-1}a_{\lambda_2}f_{\gamma}(H_S-E_{\nu_0}-E)
a^*_{\lambda_3}(H_S-E_{\nu_0}-E)^{-1})a_{\lambda_4} \nu_0 \rangle \\\nonumber
&=\sum_{\lambda_1,\lambda_2,\lambda_3,\lambda_4}\langle \phi_{\lambda_1},m_{\alpha} \rangle 
\langle m_{\gamma}, \phi_{\lambda_2} \rangle 
\langle \phi_{\lambda_3},m_{\gamma} \rangle \langle m_{\alpha},\phi_{\lambda_4}\rangle 
\frac{(1-f_{\alpha}){\tilde f}_{\gamma}(E+e_{\lambda_2}-e_{\lambda_1})}{(E-e_{\lambda_4})(E-e_{\lambda_1})}\\
&\times {\rm Tr}_{{\cal F}_S}\{\rho_S a^*_{\lambda_1}a_{\lambda_2}a^*_{\lambda_3}a_{\lambda_4} \}.
\end{align}
In the above calculations we used pull-through identities like $(H_S-z)^{-1}a_{\lambda}=a_{\lambda}(H_S-z-e_{\lambda})^{-1}$
or $a^*_{\lambda}f(H_0-z)=f(H_0-z-e_{\lambda})a^*_{\lambda}$. Now the only thing we should do is to use the Wick theorem 
for the trace in the last line (the theorem now holds as the interaction is absent):
\begin{equation}
{\rm Tr}_{{\cal F}_S}\{\rho_S a^*_{\lambda_1}a_{\lambda_2}a^*_{\lambda_3}a_{\lambda_4} \}=
n_{\lambda_1}(\delta_{\lambda_1\lambda_2}\delta_{\lambda_3\lambda_4}n_{\lambda_3} 
+\delta_{\lambda_1\lambda_4}\delta_{\lambda_2\lambda_3}  (1-n_{\lambda_2}) ),
\end{equation} 
where $n_{\lambda}={\rm Tr}_{{\cal F}}\{\rho_S a^*_{\lambda}a_{\lambda} \}$ and 
$1-n_{\lambda}={\rm Tr}_{{\cal F}}\{\rho_S a_{\lambda} a^*_{\lambda} \}$. 

The remaining terms in Eq.(\ref{J4fin}) have to be manipulated in the same manner. Collecting all 
of them one observes that all products of Fermi functions vanish, and all factors like $\tilde{f}$ will only appear as $\tilde{f}(E)$ with 
$E\in [-2t_L,2t_L]$ which allows us to drop $\chi_L$. Then the cotunneling current takes the
following form (which does not depend on the particle number $n_{\lambda}$):

\begin{equation}\label{cotLB}
\tau^4C_{\alpha,4}(0_+)=\frac{\tau^4}{\pi^2t_L^2}\sum_{\gamma}\int_{-2t_L}^{2t_L} [1-E^2/(4t_L^2)]
|\langle m_{\alpha},(h_S-E)^{-1} m_{\gamma}\rangle |^2 (f_{\alpha}(E)-f_{\gamma}(E))\; dE.
\end{equation}
One recognizes at once the 1st term in the expansion of the Landauer formula \eqref{LB} w.r.t $\tau$ in the off-resonant case. So as expected, the 
off-resonant transport is still described by a Landauer formula in the non-interacting case.
As expected, in this case the steady-state current does not depend on the initial state of the sample.

\section{Numerical simulations of the transient regime}\label{numerical}

Let us consider the same two-site system as the one in Remark \ref{remarca244}. As we have already mentioned, one can numerically 
compute transients via the generalised
Master equation (GME) method \cite{MMG}. The main idea behind this method is to write down an equation for
the reduced density operator (RDO) $\rho_r(t):={\rm Tr}_{{\cal F}_L}\{\rho\}$. Note that  $\rho_r(t)$ only acts in the Fock space of the sample. 
Its derivative w.r.t time gives the evolution of the particle number in the sample which in turn is related to the currents flowing 
to and from the leads via the continuity equation. The method is usually formulated in terms of Liouvillians (see e.g. \cite{Timm} for 
relevant equations). Although the main regime considered in other papers is the resonant one, here we pay more attention to 
the off-resonant regime. We are motivated by the fact that in our paper the transient current due to sequential tunneling 
processes is given by a rather simple analytical formula Eq.(\ref{J2trans}), which should be the main contribution on a time scale of order $1/\tau$.

Moreover, since GME also works for the resonant case, it would be a proper tool to compare the
two-regimes. The off-resonant setup is achieved by taking a small hopping constant on the leads and by globally shifting the
leads' spectrum $\sigma (h_L)=[-2t_L+E_{{\rm shift}},2t_L+E_{{\rm shift}}]$. The bias window $[\mu_R,\mu_L]$ is also fixed such that all
the many-body states of the sample are below it. The time-dependent currents in the left (L) and right (R) leads are presented in
Fig.\, 1a.

\begin{figure}
\includegraphics[width=0.5\textwidth]{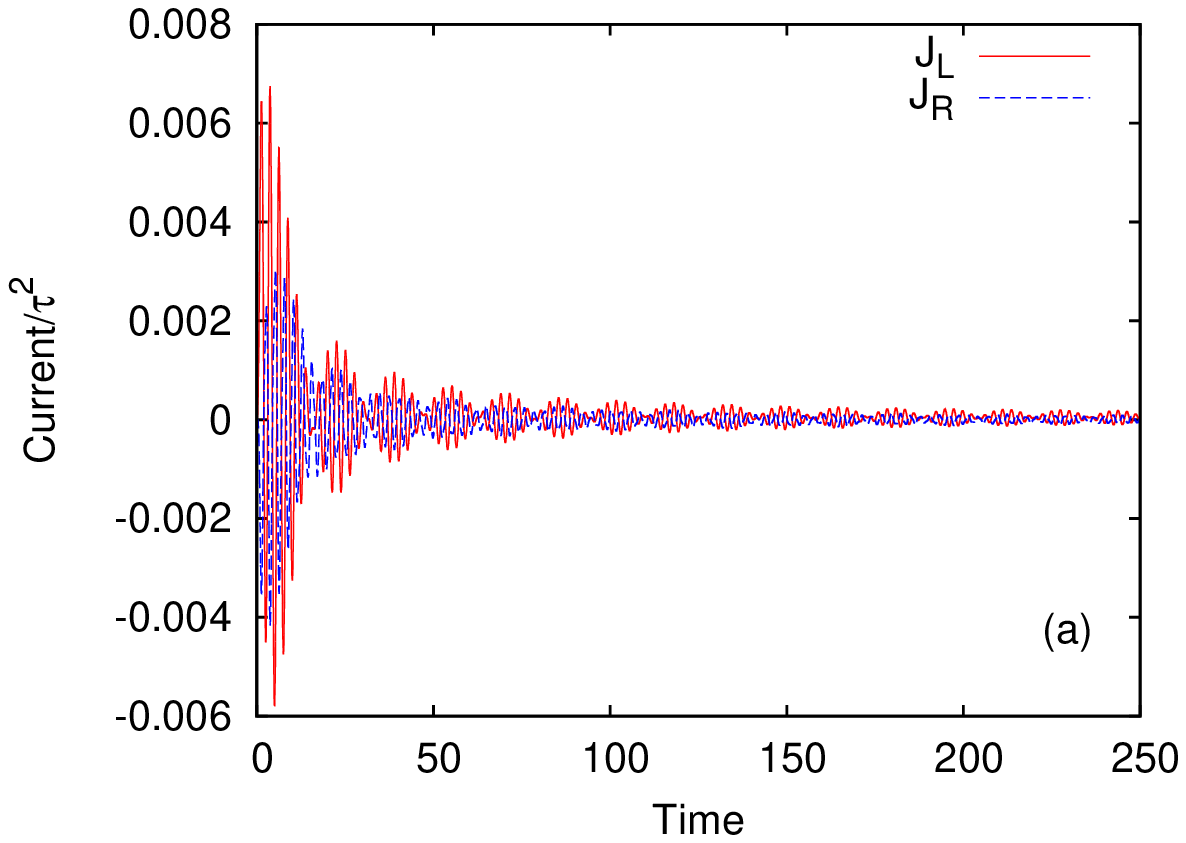}
\includegraphics[width=0.5\textwidth]{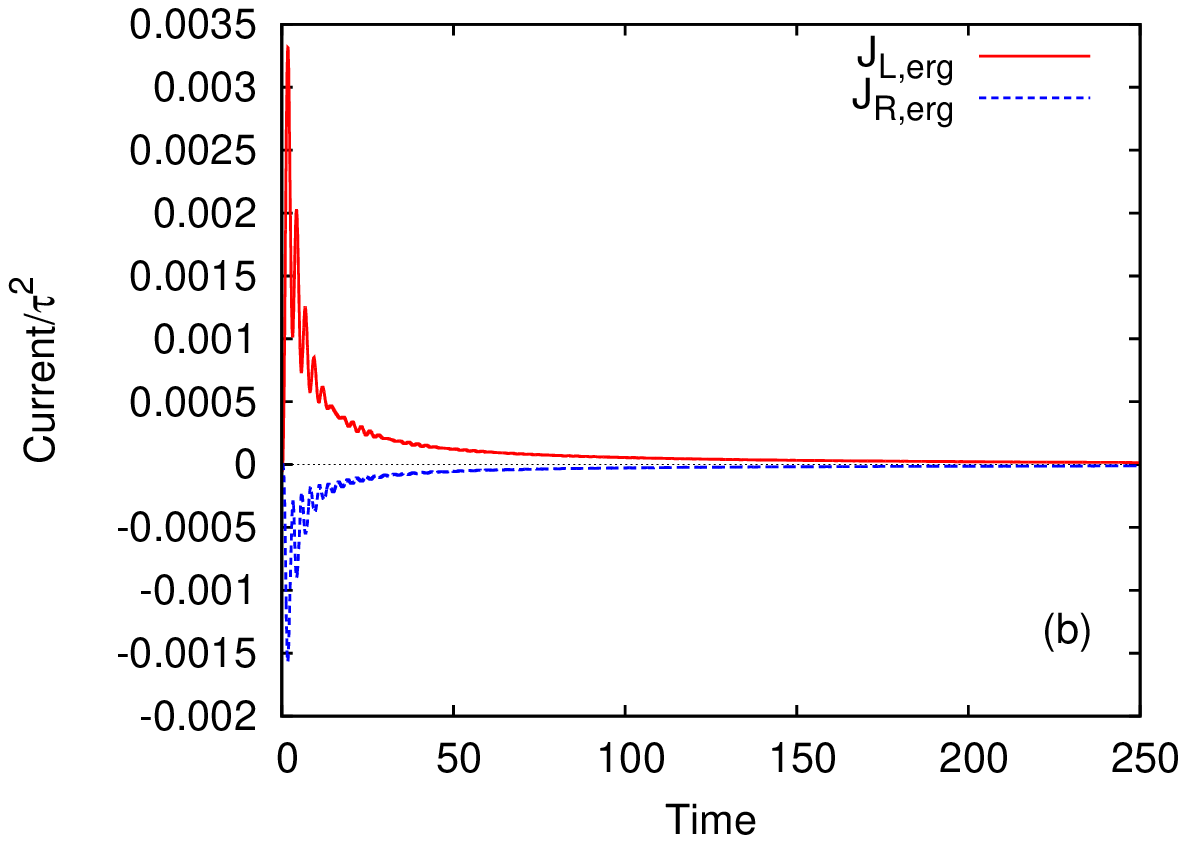}
\caption{(Color online) (a) The total transient currents $J_L$ and $J_R$ as a function of time in the off-resonant regime.
(b) The 'ergodic' currents.
Other parameters: $U=0.5$, $\tau=0.5$, $t_L=0.1$,  $E_{{\rm shift}}=6$, $\mu_L=7$, $\mu_R=6$.}
\label{fig1}
\end{figure}

The convention for the sign of the currents is as follows: $J_L$ is positive if it flows from the left lead towards the 
sample and $J_R>0$ if the current flows from the sample to the right lead. The steady-state regime thus implies 
$J_L(t)=J_R(t)$ for some $t$. Instead of this one notices that both currents exhibit modulated 
oscillations around zero and no steady-state is achieved, although the amplitude of the oscillations decreases in time. 
This behavior could be predicted by our analytical result (see Eq.(\ref{J2trans})).
However, if one performs the ergodic limit the results converges to zero in the long-time limit, as clearly seen in Fig.\,1b.

The transport in the resonant regime is shown in Fig.\,2a for two initial conditions of the isolated quantum dot 
$|\nu_0\rangle=|10\rangle $ and $|\nu_0\rangle=|00\rangle $. In this case we consider a larger $t_L$ and the bias window 
is chosen such that the first state of the dot is below it while the other ones within the bias window. Notice that 
in this case the parameters are set such that $\sigma (h_L)$ covers the entire spectrum of $h_S$. 
The transients are quite smooth and the steady state
is achieved around $t\sim 225$. In this case there is no need to consider the ergodic limit.

\begin{remark} In the resonant regime, the steady-state current {\it does not depend on the initial condition of the sample}. This has already been  
rigorously established both in the non-interacting case \cite{Aschb, CNZ} and for weakly interacting systems \cite{CMP}. 
\end{remark}

\begin{remark} If the quantum dot is initially empty, the current on the right lead starts by being negative, which means that this lead actually 
feeds as well the dot. Fig.\,2b  shows the charge that accumulates in time on the many-body states containing 
$n$ particles, and the total charge $n_{{\rm tot}}$ (the curves correspond to the initial condition $|\nu_0\rangle =|10\rangle $).
The reading of the numerical results is straightforward. The single particle state are depleted in favor of the two-particle
state $|11\rangle$. In the steady-state regime the latter contains in average one electron, because the state $|11\rangle$ contained 
within the bias window charges/discharges by back-and-forth tunneling of one electron from the leads.   
\end{remark}

\begin{figure}
\includegraphics[width=0.5\textwidth]{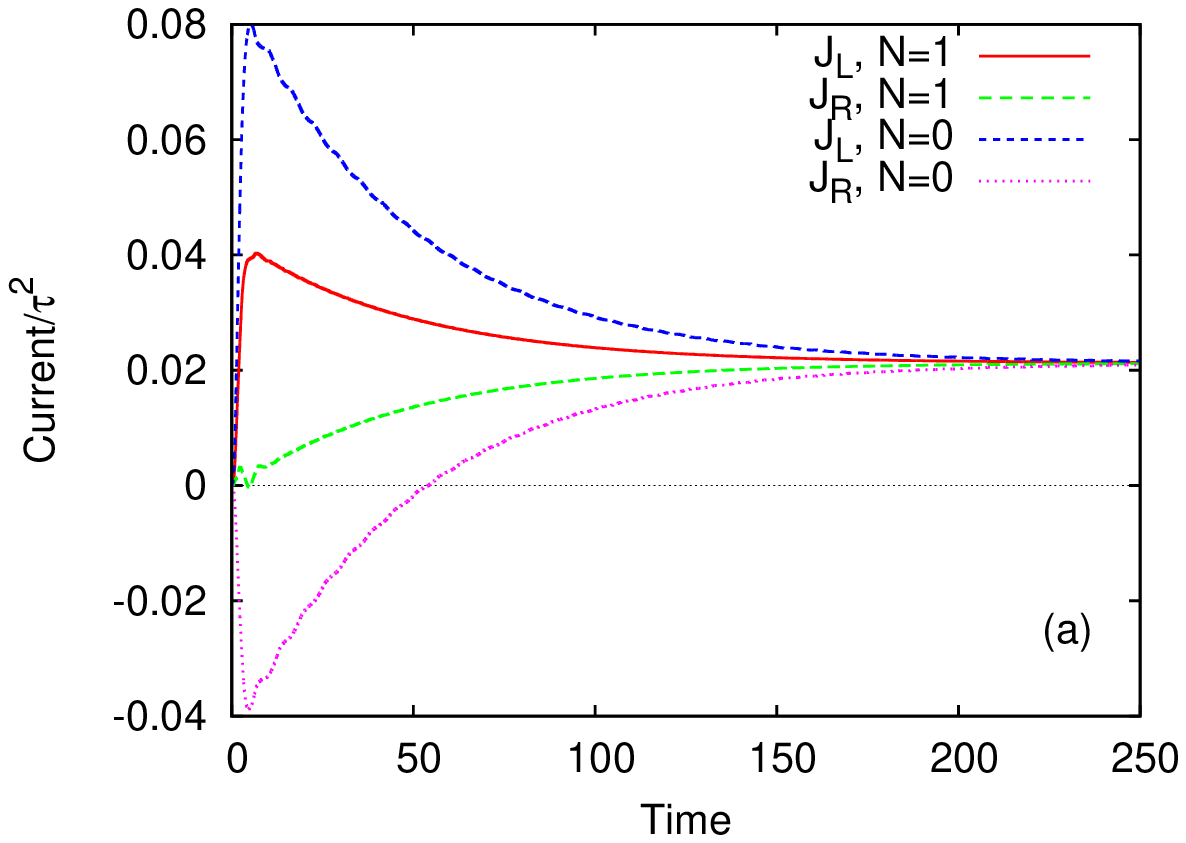}
\includegraphics[width=0.5\textwidth]{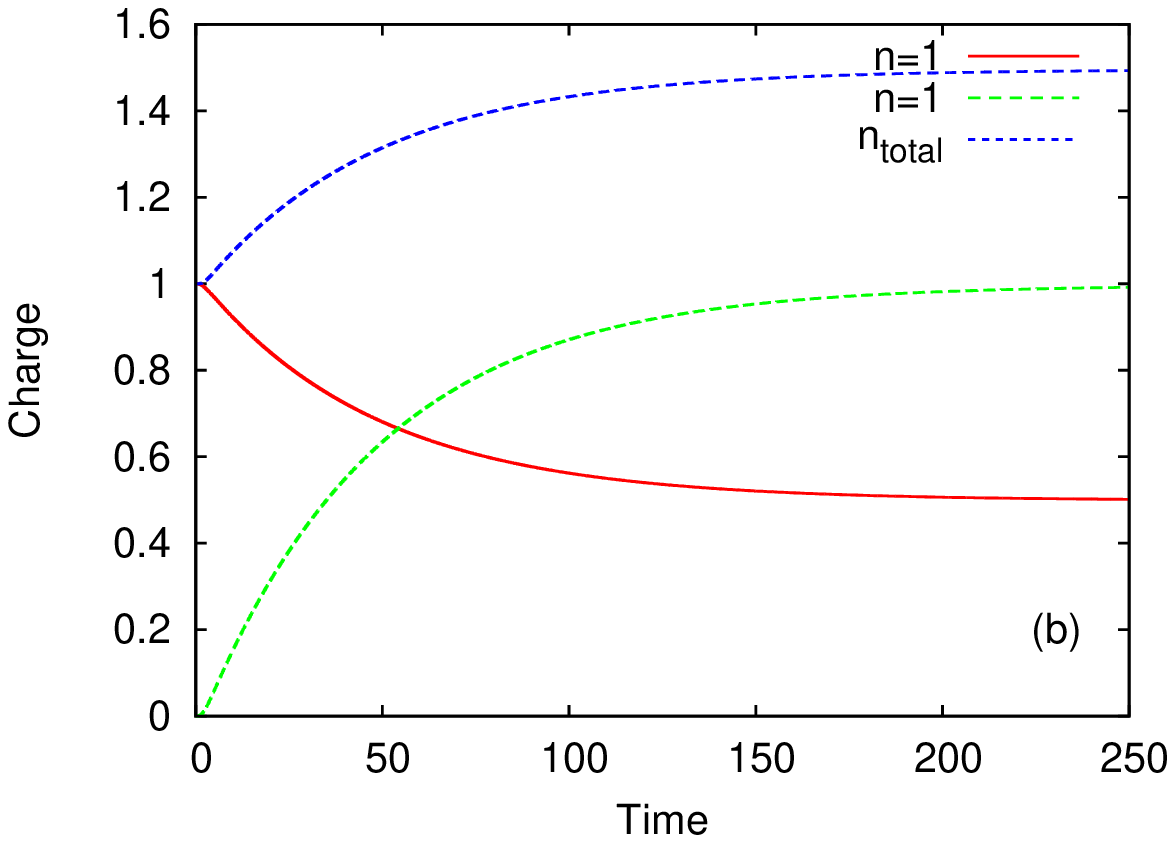}
\caption{(Color online)
(a) The total transient currents $J_L$ and $J_R$ as a function of time in the resonant regime. Two initial conditions were considered
$N=1$ corresponding to one electron on the lowest state and $N=0$ corresponding to an empty sample.
(b) The occupation of the many-body states with $n$ electrons and the total occupation.
Other parameters: $U=0.5$, $\tau=0.5$, $t_L=1.5$, $E_{{\rm shift}}=3$, $\mu_L=5$, $\mu_R=2$.}
\label{fig2}
\end{figure}

\begin{remark}
The results presented in this section were obtained by numerically implementing and solving the integro-differential equation 
for the reduced density operator which served us to calculate the transients. A legitimate question is how one could use
the GME method if interested only in the steady-state regime? The most tempting step is to assume that a 
steady-state exists, which in terms of the RDO means that $\lim_{t\to\infty}\dot\rho_r(t)=0$. If so, then one can calculate 
the stationary RDO from the GME equation and derive the steady-state currents. This strategy is extensively used in the 
physical literature. Our analysis shows that in the off-resonant regime such an approach is not justified because 
there is no steady-state. The correct procedure is to work out the time-dependent equations and calculate various 
contributions to the {\it ergodic} current which is the meaningful quantity to look at.    

\end{remark}

\section{Conclusions}\label{concluzii}

We have presented a rigorous approach to the cotunneling transport in weakly-coupled interacting quantum dots.
Using the expansion of the transient current in powers of the lead-dot coupling parameter $\tau$ we analysed
the leading order contribution (i.e. ${\cal O}(\tau^4)$) of the ergodic current which is the relevant 
quantity to consider in this regime. Explicit calculations for elastic and inelastic cotunneling contributions 
to transport were presented. For non-interacting electrons one recovers the Landauer formula. For a simple 
two-level system, we show that in the interacting case the cotunneling current depends on the initial many-body 
configuration of the dot. To our best knowledge,  
this memory-effect has not been reported before. An explicit formula for the ergodic sequential tunneling 
current (i.e. ${\cal O}(\tau^2)$) is given. This contribution vanishes in the cotunneling regime but
the transient sequential tunneling does not reach a stationary state. These results are also recovered
through numerical simulations via the generalized Master equation method. This method allows calculation 
of transient sequential tunneling currents. A generalized Master equation containing higher order terms  
has been recently reported \cite{Timm1},\cite{Koller}. This motivates a thorough rigorous analysis
 on the existence of a stationary regime for the reduced density operator. 

\vspace{0.5cm}

{\bf Acknowledgements}. Both authors acknowledge support from the Danish FNU grant {\it Mathematical Physics}.
V.\,M.\ acknowledges the financial support from PNCDI2 program (grant No.\ 515/2009),
Core Project (grant No.\ 45N/2009).

\end{document}